\documentclass[12pt]{article}

\usepackage{amssymb,amsthm,amscd,array}
\usepackage{graphicx}
\usepackage[final]{epsfig}

\let\ssection=\section
\renewcommand{\section}{\setcounter{equation}{0}\ssection}

\chardef\s=110 \chardef\g=103

\newtheorem{thm}{Theorem}[section]
\newtheorem{lem}[thm]{Lemma}

\newtheorem{prop}[thm]{Proposition}
\newtheorem{rmk}[thm]{Remark}

\def\g{\gamma}

\def\s{\sigma}

\begin{document}

\title{Projectively Equivariant Quantization and Symbol calculus in dimension $1|2$ }
\author{N. Mellouli \thanks{
Universit\'e Lyon 1, Institut Camille Jordan, CNRS, UMR5208, Villeurbanne,
F-69622, France; mellouli@math.univ-lyon1.fr}}
\date{23/06/2011}
\maketitle

\begin{abstract}
The spaces of higher-order differential operators (in Dimension $1|2$),
which are modules over the stringy Lie superalgebra $\mathcal{K}(2) $, are
isomorphic to the corresponding spaces of symbols as orthosymplectic modules
in non resonant cases. Such an $\mathrm{osp}\left( 2|2\right) $-equivariant
quantization, which has been given in second-order differential operators
case, keeps existing and unique. We calculate its explicit formula that
provides extension in particular order cases.
\end{abstract}

\thispagestyle{empty}

\section{Introduction and the main results}

Let $S^{1|2}$ be a supermanifold which is endowed with a projective
structure (Susy-structure), see \cite{NV}, in dimension $1|2$ and $\mathcal{D%
}_{\lambda ,\mu }\left( S^{1|2}\right) $ the space of differential operators
on $S^{1|2}$ acting from the space of $\lambda $-densities to the space of $%
\mu $-densities where $\lambda $ and~$\mu $ are real or complex numbers. The
space $\mathcal{D}_{\lambda ,\mu }\left( S^{1|2}\right) $ which is a module
over the stringy superalgebra $\mathcal{K}(2)$, see \cite{GLS1}, is
naturally filtrated and has an other finer filtration by the contact order
of differential operators. The space of symbols $\mathcal{S}\left(
S^{1|2}\right) $, that is the graded module $\mathrm{gr}\mathcal{D}_{\lambda
\mu }\left( S^{1|2}\right) $, isn't isomorphic to the space $\mathcal{D}%
_{\lambda ,\mu }\left( S^{1|2}\right) $ as $\mathcal{K}(2)$-module.
Therefore, we have restricted the module structure on $\mathcal{D}_{\lambda
\mu }\left( S^{1|2}\right) $ to the orthosymplectic Lie superalgebra $%
\mathrm{osp}\left( 2|2\right) $ that is naturally embedded into $\mathcal{K}%
(2)$. We establish a canonical isomorphism between the space of differential
operators on $S^{1|2}$ and the corresponding space of symbols. An explicit
expression of projectively equivariant quantization map is given in case of
second order differential operetors, see \cite{NV}. We extand calculus to
symbols of higher order differential operators.

\noindent

\section{Geometry of the supercircle $S^{1|2}$}

We have considered the supercircle $S^{1|2}$ described in \cite{NV} by its
graded commutative algebra of complex-valued functions $C^{\infty }\left(
S^{1|2}\right) $, consisting of the following elements : 
\begin{equation}
f\left( x,\xi _{1},\xi _{2}\right) =\ f_{0}(x)+\xi _{1}\,f_{1}(x)+\xi
_{2}\,f_{2}(x)+\xi _{1}\xi _{2}\,f_{12}(x),  \label{ResValSh}
\end{equation}
where $x$ is the Fourier image of the angle parameter on $S^{1}$, $\xi
_{1},\xi _{2}$ are odd Grassmann coordinates and $f_{0},f_{12},f_{1},f_{2}%
\in {}C^{\infty }(S)$ are functions with complex values. We have defined the
parity function~$p$ by setting $p\left( x\right) =0$ and $p\left( \xi
_{i}\right) =1$.

The \textit{standard contact} structure on $S^{1|2}$, known as \textit{Susy-}%
structure\textit{,} is defined by the data of a linear distribution $%
\left\langle \overline{D}_{1},\overline{D}_{2}\right\rangle $ on $S^{1|2}$
generated by the odd vector fields : 
\begin{equation}
\overline{D}_{1}=\partial _{\xi _{1}}-\xi _{1}\partial _{x},\qquad \overline{%
D}_{2}=\partial _{\xi _{2}}-\xi _{2}\partial _{x}.  \label{ConVF}
\end{equation}
We would rather recall that every contact vector field can be expressed, for
some function $f\in C^{\infty }\left( S^{1|2}\right) $, by 
\begin{equation}
X_{f}=f\partial _{x}-\left( -1\right) ^{p\left( f\right) }\frac{1}{2}\left( 
\overline{D}_{1}\left( f\right) \overline{D}_{1}\ +\overline{D}_{2}\left(
f\right) \overline{D}_{2}\right)  \label{ContVF}
\end{equation}
The \textit{projective (conformal)} structure on the supercircle $S^{1|2}$,
see \cite{OOC}, is defined by the action of Lie superalgebra $\mathrm{osp}%
\left( 2|2\right) $. The \textit{orthosymplectic} algebra $\mathrm{osp}%
\left( 2|2\right) $ is spanned by the contact vector fields $X_{f}$ with the
contact Hamiltonians $f$ which are elements of $\left\{ 1,\xi _{1},\xi
_{2},x,\xi _{1}\xi _{2},x\xi _{1},x\xi _{2},x^{2}\right\} $. The embedding
of $\mathrm{osp}\left( 2|2\right) $ into $\mathcal{K}\left( 2\right) $ is
given by (\ref{ContVF}). The subalgebra $\mathrm{Aff}\left( 2|2\right) $ of $%
\mathrm{osp}\left( 2|2\right) $, called the \textit{Affine} Lie
superalgebra, is spanned by the contact vector fields $X_{f}$ with the
contact Hamiltonians $f$ which are elements of $\left\{ 1,\xi _{1},\xi
_{2},x,\xi _{1}\xi _{2}\right\} $.

For any contact vector field, we have defined a family of \ differential
operators of order one on $C^{\infty }\left( S^{1|2}\right) :$ 
\begin{equation}
L_{X_{f}}^{\lambda }:=X_{f}+\lambda f^{\prime },
\end{equation}
where the parameter $\lambda $ is an arbitrary complex number. Thus, we have
obtained a family of $\mathcal{K}\left( 2\right) $-modules on $C^{\infty
}\left( S^{1|2}\right) $ noted by $\mathcal{F}_{\lambda }\left(
S^{1|2}\right) $ which are called the spaces of \textit{\ weighted densities}
of weight $\lambda $.

\section{Differential operators on the spaces of weighted densities}

In this section, we have introduced the space of differential operators
acting on the spaces of weighted densities. We have also presented the
corresponding space of symbols on $S^{1|2}$. Those are detailed in \cite
{LKV,GLS,GMO,Con,NV}.

For every integer or half-integer $k$, the space of differential operators
of the form 
\begin{equation}
A=\sum_{\ell +\frac{m}{2}+\frac{n}{2}\leq {}k}a_{\ell ,m,n}\partial
_{x}^{\ell }\,\overline{D}_{1}^{m}\,\overline{D}_{2}^{n},
\label{GeneralformulaOpDiff}
\end{equation}
where $a_{\ell ,m,n}\in C^{\infty }\left( S^{1|2}\right) $ and $m,n\leq 1$,
has been noted by $\mathcal{D}_{\lambda \mu }^{k}\left( S^{1|2}\right) $.
This above $\mathcal{K}(2)$-module space has a $\mathcal{K}(2)$-invariant 
\textit{finer filtration} : 
\begin{equation}
\mathcal{D}_{\lambda \mu }^{0}\left( S^{1|2}\right) \subset \mathcal{D}%
_{\lambda \mu }^{\frac{1}{2}}\left( S^{1|2}\right) \subset \mathcal{D}%
_{\lambda \mu }^{1}\left( S^{1|2}\right) \subset ...\subset \mathcal{D}%
_{\lambda \mu }^{k}\left( S^{1|2}\right) \subset \mathcal{D}_{\lambda \mu
}^{k+\frac{1}{2}}\left( S^{1|2}\right) \subset \cdots  \label{finefiltration}
\end{equation}
that has been considered in papers \cite{NV,GMO}.

We would rather remind that the orthosymplectic superalgebra $\mathrm{osp}%
\left( 2|2\right) $ has been acting on $\mathcal{D}_{\lambda \mu }^{k}\left(
S^{1|2}\right) $. The action of contact field $X_{f}$ of order one on a
differential operator $A$, see formula (\ref{ContVF}), is given by 
\begin{equation}
\pounds _{X_{f}}^{\lambda \mu }\left( A\right) =L_{X_{f}}^{\mu }\circ
A-\left( -1\right) ^{p\left( f\right) p\left( A\right) }A\circ
L_{X_{f}}^{\lambda }.  \label{GenLD}
\end{equation}

\subsection{Space of symbols of differential operators}

The graded $\mathcal{K}\left( 2\right) $-module, which is associated with
the finer filtration (\ref{finefiltration}) and called the \textit{space of
symbols} of differential operators, is defined by 
\begin{equation}
\mathrm{gr}\mathcal{D}_{\lambda \mu }\left( S^{1|2}\right)
=\bigoplus_{i=0}^{\infty }\mathrm{gr}^{\frac{i}{2}}\mathcal{D}_{\lambda \mu
}\left( S^{1|2}\right) ,  \label{DefSymb}
\end{equation}
where $\mathrm{gr}^{k}\,\mathcal{D}_{\lambda \mu }\left( S^{1|2}\right) =%
\mathcal{D}_{\lambda \mu }^{k}\left( S^{1|2}\right) /\mathcal{D}_{\lambda
\mu }^{k-\frac{1}{2}}\left( S^{1|2}\right) $ for every integer or
half-integer $k.$

The image of any differential operator through the natural projection 
\[
\sigma _{pr}:\mathcal{D}_{\lambda \mu }^{k}\left( S^{1|2}\right) \rightarrow 
\mathrm{gr}^{k}\,\mathcal{D}_{\lambda \mu }\left( S^{1|2}\right) , 
\]
that is defined by the filtration (\ref{finefiltration}), has been called
the \textit{principal symbol}.

Referring to \cite{NV}, the \ stringy superalgebra $\mathcal{K}\left(
2\right) $ keeps acting on the space of symbols.

\begin{prop}
\label{SymP} If $k$ is an integer, then 
\begin{equation}
\mathrm{gr}^k\mathcal{D}_{\lambda \mu}\left( S^{1|2}\right) =\mathcal{F}
_{\mu -\lambda -k} \bigoplus \mathcal{F}_{\mu -\lambda -k}
\end{equation}
\end{prop}

\begin{proof}
By definition (see formula (\ref{GeneralformulaOpDiff})), a given operator $%
A\in \mathcal{D}_{\lambda \mu }^{k}\left( S^{1|2}\right) $ with integer $k$
is of the form 
\[
A=F_{1}\,\partial _{x}^{k}+F_{2}\,\partial _{x}^{k-1}\overline{D}_{1}%
\overline{D}_{2}+\cdots , 
\]
where $\cdots $ stand for lower order terms. The principal symbol of $A$ is
then encoded by the pair $(F_{1},F_{2})$. From (\ref{GenLD}), we easily
calculate the $\mathcal{K}(2)$-action on the principal symbol: 
\[
L_{X_{f}}\left( F_{1},F_{2}\right) =\left( L_{X_{f}}^{\mu -\lambda -k}\left(
F_{1}\right) ,L_{X_{f}}^{\mu -\lambda -k}\left( F_{2}\right) .\right) 
\]
In other words, both $F_{1}$ and $F_{2}$ transform as $(\mu -\lambda -k)$%
-densities. 
\end{proof}

The situation is more complicated for half-integer $k$ : the $\mathcal{K}%
\left( 2\right) $-action has been given by 
\begin{equation}
L_{X_{f}}(F_{1},F_{2})=(L_{X_{f}}^{\mu -\lambda -k}\left( F_{1}\right) -%
\frac{1}{2}\,\overline{D}_{1}\,\overline{D}_{2}\left( f\right)
F_{2},\;L_{X_{f}}^{\mu -\lambda -k}\left( F_{2}\right) +\frac{1}{2}\overline{%
D}_{1}\,\overline{D}_{2}\left( f\right) F_{1}).  \label{Sact}
\end{equation}
Therefore, the spaces of symbols of half-integer contact order aren't
isomorphic to the spaces of weighted densities.

\noindent Simplifying the notation as in \cite{LO,DLO,GMO,NV}, the whole
space of symbols $\mathrm{gr}\mathcal{D}_{\lambda \mu }\left( S^{1|2}\right) 
$, depending only on $\mu -\lambda $, has been noted by $\mathcal{S}_{\mu
-\lambda }\left( S^{1|2}\right) $, and the space of symbols of contact order 
$k$ has been noted by $\mathcal{S}_{\mu -\lambda }^{k}\left( S^{1|2}\right).$

A linear map, $Q:\mathcal{S}_{\mu -\lambda }\left( S^{1|2}\right)
\rightarrow \mathcal{D}_{\lambda \mu }\left( S^{1|2}\right) $, is called a 
\textit{quantization map} if it verifies bijectivety and preserves the
principal symbol of every differential operator.

\section{Projectively equivariant quantization on $S^{1|2}$}

The main result of this paper is the existence and uniqueness of an $\mathrm{%
osp}\left( 2|2\right) $-equivariant quantization map in Dimension $1|2$. We
calculate its explicit formula.

For every $m$ integer or half-integer, the space $\mathcal{D}_{\lambda \mu
}^{m}\left( S^{1|2}\right) $ is isomorphic to the corresponding space of
symbols as an $\mathrm{Aff}\left( 2|2\right) $-module. We will show how to
extend this isomorphism to that of the $\mathrm{osp}\left( 2|2\right) $%
-modules.

\subsection{The Divergence operators as Affine equivariant}

Let us introduce new differential operators, which are called Divergence
operators, on the space of symbols $\mathcal{S}_{\mu -\lambda }\left(
S^{1|2}\right) $.

At first, we consider the case of differential operators of contact order $k$%
, where $k$\ \ is an integer. We have assumed that the symbols of
differential operators are homogeneous and we have defined parity of the non
vanished symbol $\left( F_{1},F_{2}\right) $ as $p\left( F\right) :=p\left(
F_{1}\right) =p\left( F_{2}\right) $.

\subsubsection{The Divergence operators in case of integer contact order $k$}

In this case, we define the Divergence as Affine equivariant differential
operators on the space of symbols $\mathcal{S}_{\mu -\lambda }\left(
S^{1|2}\right) $. In each component $\mathcal{S}_{\mu -\lambda }^{k}\left(
S^{1|2}\right) $, we have 
\begin{equation}
DIV^{2n+1}\left( F_{1},F_{2}\right) =\left( -1\right) ^{p\left( F\right)
+1}\left( 
\begin{array}{c}
\frac{k+2\lambda }{k}\partial _{x}^{n}\overline{D}_{2}\left( F_{2}\right)
+\partial _{x}^{n}\overline{D}_{1}\left( F_{1}\right) \\ 
\partial _{x}^{n}\overline{D}_{2}\left( F_{1}\right) -\frac{k+2\lambda }{k}%
\partial _{x}^{n}\overline{D}_{1}\left( F_{2}\right)
\end{array}
\right) ,  \label{DivergenceEven1}
\end{equation}
\begin{equation}
DIV^{2n}\left( F_{1},F_{2}\right) =\left( 
\begin{array}{c}
\partial _{x}^{n}\left( F_{1}\right) -\frac{\left( k+2\lambda \right) n}{%
k\left( 2\left( \mu -\lambda \right) +n-2k\right) }\partial _{x}^{n-1}%
\overline{D}_{1}\overline{D}_{2}\left( F_{2}\right) \\ 
\frac{\left( k+2\lambda \right) \left( k-n\right) }{k\left( k-n+2\lambda
\right) }\partial _{x}^{n}\left( F_{2}\right) +\frac{n\left( k-n\right) }{%
\left( 2\left( \mu -\lambda \right) +n-2k\right) \left( k-n+2\lambda \right) 
}\partial _{x}^{n-1}\overline{D}_{1}\overline{D}_{2}\left( F_{1}\right)
\end{array}
\right)  \label{DivergenceEven2}
\end{equation}
and 
\begin{eqnarray*}
div^{2k-\left( 2n+1\right) } &=&\left( 
\begin{array}{cc}
\partial _{x}^{k-n-1}\overline{D}_{1}, & \partial _{x}^{k-n-1}\overline{D}%
_{2}
\end{array}
\right) , \\
div^{2k-\left( 2n\right) } &=&\left( 
\begin{array}{cc}
\partial _{x}^{k-n}, & \partial _{x}^{k-n-1}\overline{D}_{1}\overline{D}_{2}
\end{array}
\right) .
\end{eqnarray*}

\begin{lem}
The Divergence operators (\ref{DivergenceEven1}) and (\ref{DivergenceEven2})
commute with the $\mathrm{Aff}\left( 2|2\right) $-action.
\end{lem}

\begin{proof}
This is a direct consequence of projectively equivariant symbol calculus.

We are looking for the symbols:
\[
DIV^{2n+1}\left( F_{1},F_{2}\right) =\left( 
\begin{array}{c}
c_{1}\partial _{x}^{n}\overline{D}_{2}\left( F_{2}\right) +c_{2}\partial
_{x}^{n}\overline{D}_{1}\left( F_{1}\right)  \\ 
c_{3}\partial _{x}^{n}\overline{D}_{2}\left( F_{1}\right) +c_{4}\partial
_{x}^{n}\overline{D}_{1}\left( F_{2}\right) 
\end{array}
\right) 
\]
and 
\[
DIV^{2n}\left( F_{1},F_{2}\right) =\left( 
\begin{array}{c}
c_{5}\partial _{x}^{n}\left( F_{1}\right) +c_{6}\partial _{x}^{n-1}\overline{%
D}_{1}\overline{D}_{2}\left( F_{2}\right)  \\ 
c_{7}\partial _{x}^{n}\left( F_{2}\right) +c_{8}\partial _{x}^{n-1}\overline{%
D}_{1}\overline{D}_{2}\left( F_{1}\right) 
\end{array}
\right) 
\]
where $c_{i}(1\leq i\leq 8)$ are arbitrary constants. From the commutation
relation $\left[ L_{X_{f}},DIV\right] $ for $f\in \mathrm{Aff}\left(
2|2\right) $ we easily get the $\mathrm{Aff}\left( 2|2\right) $-equivariance
if Divergence operators .
\end{proof}

\subsubsection{The Divergence operators in case of half-integer contact
order $k+\frac{1}{2}$}

In this case, we also define the Divergence as Affine equivariant
differential operators on the space of symbols $\mathcal{S}_{\mu -\lambda
}\left( S^{1|2}\right) $. In each component $\mathcal{S}_{\mu -\lambda }^{k+%
\frac{1}{2}}\left( S^{1|2}\right) $ we have 
\begin{equation}
DIV^{2n+1}\left( F_{1},F_{2}\right) =\left( -1\right) ^{p\left( F\right) }%
\frac{2\left( \mu -\lambda \right) -\left( 2k+1\right) }{2\left( \mu
-\lambda \right) -2k}\left( 
\begin{array}{c}
\partial _{x}^{n}\overline{D}_{2}\left( F_{2}\right) +\partial _{x}^{n}%
\overline{D}_{1}\left( F_{1}\right) \\ 
\frac{k-n}{k-n+2\lambda }\left( \partial _{x}^{n}\overline{D}_{2}\left(
F_{1}\right) -\partial _{x}^{n}\overline{D}_{1}\left( F_{2}\right) \right)
\end{array}
\right) ,  \label{DivergenceOdd1}
\end{equation}
\begin{equation}
DIV^{2n}\left( F_{1},F_{2}\right) =\frac{2\left( \mu -\lambda \right)
-\left( 2k+1\right) }{2\left( \mu -\lambda \right) -2k}\left( 
\begin{array}{c}
\frac{2\left( \mu -\lambda \right) +n-2k}{2\left( \mu -\lambda \right)
+n-\left( 2k+1\right) }\partial _{x}^{n}\left( F_{1}\right) -\frac{n}{%
2\left( \mu -\lambda \right) +n-\left( 2k+1\right) }\partial _{x}^{n-1}%
\overline{D}_{1}\overline{D}_{2}\left( F_{2}\right) \\ 
\frac{2\left( \mu -\lambda \right) +n-2k}{2\left( \mu -\lambda \right)
+n-\left( 2k+1\right) }\partial _{x}^{n}\left( F_{2}\right) +\frac{n}{%
2\left( \mu -\lambda \right) +n-\left( 2k+1\right) }\partial _{x}^{n-1}%
\overline{D}_{1}\overline{D}_{2}\left( F_{1}\right)
\end{array}
\right) ,  \label{DivergenceOdd2}
\end{equation}
and

\begin{eqnarray*}
div^{2k+1-\left( 2n+1\right) } &=&\left( 
\begin{array}{cc}
\partial _{x}^{k-n}, & \partial _{x}^{k-n-1}\overline{D}_{1}\overline{D}_{2}
\end{array}
\right) , \\
div^{2k+1-\left( 2n\right) } &=&\left( 
\begin{array}{cc}
\partial _{x}^{k-n}\overline{D}_{1}, & \partial _{x}^{k-n}\overline{D}_{2}
\end{array}
\right) .
\end{eqnarray*}

\begin{lem}
The Divergence operators (\ref{DivergenceOdd1}) and (\ref{DivergenceOdd2})
commute with the action of Affine Lie superalgebra.
\end{lem}

\begin{proof}
Straightforward calculus.
\end{proof}

\subsection{Staitement of the main result}

Let us give the explicit formula of the projectively equivariant
quantization map. We will give the proof in the next section.

\begin{thm}
The unique $\mathrm{osp}\left( 2|2\right) $-equivariant quantization map
associates the following differential operator with a symbol $\left(
F_{1},F_{2}\right) \in \mathcal{S}_{\mu -\lambda }^{k}\left( S^{1|2}\right) $
where $k$\ is (even or odd) integer : 
\begin{equation}
Q\left( F_{1},F_{2}\right) =\sum_{n=0}^{k}\frac{\left( 
\begin{array}{c}
\left[ \frac{k}{2}\right] \\ 
\left[ \frac{2n+1+\left( -1\right) ^{k}}{4}\right]
\end{array}
\right) \left( 
\begin{array}{c}
\left[ \frac{k-1}{2}\right] +2\lambda \\ 
\left[ \frac{2n+1-\left( -1\right) ^{k}}{4}\right]
\end{array}
\right) }{\left( 
\begin{array}{c}
k-2\left( \mu -\lambda \right) \\ 
\left[ \frac{n+1}{2}\right]
\end{array}
\right) }DIV^{n}\left( F_{1},F_{2}\right) div^{k-n}
\label{StringQuantization}
\end{equation}
provided $\mu -\lambda \neq 0,\frac{1}{2},1,\frac{3}{2},2...$where $DIV$ and 
$div$ are defined in each particular case of even or odd contact order : (%
\ref{DivergenceEven1}), (\ref{DivergenceEven2}), (\ref{DivergenceOdd1}) and (%
\ref{DivergenceOdd2}) and the coefficients are $\left( 
\begin{array}{c}
n \\ 
m
\end{array}
\right) =\frac{n\left( n-1\right) ..\left( n-m+1\right) }{m!}.$
\end{thm}

\begin{rmk}
This theorem keeps being achieved in the case of Dimension $1|1$, see \cite
{GMO}; the divergence operators $DIV$ and $div$ are given by $\overline{D}.$
\end{rmk}

\subsection{\noindent Proof of theorem in case of $k$-differential operators}

\begin{proof}
Let us first consider the case of $k$-differential operators where $k$ is
integer. The quantization map \ref{StringQuantization} is, indeed, $\mathrm{%
osp}\left( 2|2\right) $-equivariant. Now, we are considering a
differentiable linear map $Q:\mathcal{S}_{\mu -\lambda }^{k}\left(
S^{1|2}\right) \rightarrow \mathcal{D}_{\lambda \mu }^{k}\left(
S^{1|2}\right) $ for $k\geq 1$, preserving the principal symbol. Such a map
is of the form :

\begin{eqnarray*}
Q\left( F_{1},F_{2}\right) &=&F_{1}\partial _{x}^{k}+F_{2}\partial _{x}^{k-1}%
\overline{D}_{1}\overline{D}_{2}+... \\
&&+\widetilde{Q}_{1}^{\left( 2\ell \right) }\left( F_{1}\right) +\widetilde{Q%
}_{1}^{\left( 2\ell +1\right) }\left( F_{1}\right) \\
&&+\widetilde{Q}_{2}^{\left( 2\ell \right) }\left( F_{2}\right) +\widetilde{Q%
}_{2}^{\left( 2\ell +1\right) }\left( F_{2}\right) \\
&&..+\left( C_{2k,1}\partial _{x}^{k}\left( F_{1}\right) +C_{2k,2}\partial
_{x}^{k-1}\overline{D}_{1}\overline{D}_{2}\left( F_{2}\right) \right)
\end{eqnarray*}
where $\widetilde{Q}_{1}^{\left( m\right) }$ and $\widetilde{Q}_{2}^{\left(
m\right) }$ are differential operators with coefficients in $\mathcal{F}%
_{\mu -\lambda }\left( S^{1|2}\right) $, see (\ref{GeneralformulaOpDiff}).

We obtain the following :

a) This map commutes with the action of the vector fields $D_{1},D_{2}\in 
\mathrm{osp}\left( 2|2\right) $, where $D_{i}=\partial _{\xi _{i}}+\xi
_{i}\partial _{x}$, if and only if the differential operators $\widetilde{Q}%
_{1}^{\left( m\right) }$ and $\widetilde{Q}_{2}^{\left( m\right) }$ are with
constant coefficients.

b) This map commutes with the linear vector fields $X_{\xi _{1}},X_{\xi
_{2}},X_{x}$ if and only if the differential operators $\widetilde{Q}%
_{1}^{\left( m\right) }$ and $\widetilde{Q}_{2}^{\left( m\right) }$ are of
contact order $\frac{m}{2}$ in addition to the form 
\[
\left\{ 
\begin{array}{rcl}
\widetilde{Q}_{1}^{\left( 2\ell +1\right) }\left( F_{1}\right)  & = & 
C_{2\ell +1,1}\partial _{x}^{\ell }\overline{D}_{1}\left( F_{1}\right)
\partial _{x}^{k-\ell -1}\overline{D}_{1}+C_{2\ell +1,3}\partial _{x}^{\ell }%
\overline{D}_{2}\left( F_{1}\right) \partial _{x}^{k-\ell -1}\overline{D}_{2}
\\[4pt]
\widetilde{Q}_{1}^{\left( 2\ell \right) }\left( F_{1}\right)  & = & C_{2\ell
,1}\partial _{x}^{\ell }\left( F_{1}\right) \partial _{x}^{k-\ell }+C_{2\ell
,3}\partial _{x}^{\ell -1}\overline{D}_{1}\overline{D}_{2}\left(
F_{1}\right) \partial _{x}^{k-\ell -1}\overline{D}_{1}\overline{D}_{2} \\%
[4pt]
\widetilde{Q}_{2}^{\left( 2\ell +1\right) }\left( F_{2}\right)  & = & 
C_{2\ell +1,2}\partial _{x}^{\ell }\overline{D}_{2}\left( F_{2}\right)
\partial _{x}^{k-\ell -1}\overline{D}_{1}+C_{2\ell +1,4}\partial _{x}^{\ell }%
\overline{D}_{1}\left( F_{2}\right) \partial _{x}^{k-\ell -1}\overline{D}_{2}
\\[4pt]
\widetilde{Q}_{2}^{\left( 2\ell \right) }\left( F_{2}\right)  & = & C_{2\ell
,2}\partial _{x}^{\ell -1}\overline{D}_{1}\overline{D}_{2}\left(
F_{2}\right) \partial _{x}^{k-\ell }+C_{2\ell ,4}\partial _{x}^{\ell }\left(
F_{2}\right) \partial _{x}^{k-\ell -1}\overline{D}_{1}\overline{D}_{2}
\end{array}
\right. 
\]
where the coefficients $C_{m,i}(i=1,2,3,4)$ are arbitrary constants.

Note that the vector field $X_{x\xi _{i}}$ is the commutation relation $%
\left[ X_{\xi _{i}},X_{x^{2}}\right] ,i=1,2,$ so it is sufficient to impose
the equivariance with respect to the vector field $X_{x^{2}}$ to meet the
whole condition of $\mathrm{osp}\left( 2|2\right) $-equivariance.

d) The above quantization map commutes with the action of $X_{x^{2}}$ if and
only if any the coefficients $C_{m,i}$ verify the following conditions :

\[
\left\{ 
\begin{array}{rcl}
\ell \left( \ell -1+2\left( \mu -\lambda -k\right) \right) C_{2\ell ,1} & =
& -\left( k-\ell +1\right) \left( k-\ell +2\lambda \right) C_{2\ell -2,1} \\%
[4pt] 
\left( \ell +1\right) \left( \ell +2\left( \mu -\lambda -k\right) \right)
C_{2\ell ,2} & = & \left( k-\ell +2\lambda \right) \left( 
\begin{array}{c}
\left( -1\right) ^{p\left( F\right) }\left( C_{2\ell -1,2}-C_{2\ell
-1,4}\right) \\ 
-\left( k-\ell +1\right) C_{2\ell -2,2}
\end{array}
\right) \\[4pt] 
\left( \ell +1\right) \left( \ell +2\left( \mu -\lambda -k\right) \right)
C_{2\ell ,3} & = & -\left( k-\ell \right) \left( 
\begin{array}{c}
\left( -1\right) ^{p\left( F\right) }\left( C_{2\ell -1,1}+C_{2\ell
-1,3}\right) \\ 
+\left( k-\ell +2\lambda +1\right) C_{2\ell -2,3}
\end{array}
\right) \\[4pt] 
\ell \left( \ell -1+2\left( \mu -\lambda -k\right) \right) C_{2\ell ,4} & =
& -\left( k-\ell \right) \left( k-\ell +2\lambda +1\right) C_{2\ell -2,4} \\%
[4pt] 
\left( \ell +1\right) \left( \ell +2\left( \mu -\lambda -k\right) \right)
C_{2\ell +1,1} & = & \left( k-\ell \right) \left( \left( -1\right) ^{p\left(
F\right) }C_{2\ell ,1}-\left( k-\ell +2\lambda \right) C_{2\ell -1,1}\right)
\\[4pt] 
\left( \ell +1\right) \left( \ell +2\left( \mu -\lambda -k\right) \right)
C_{2\ell +1,2} & = & \left( k-\ell +2\lambda \right) \left( \left( -1\right)
^{p\left( F\right) }C_{2\ell ,4}-\left( k-\ell \right) C_{2\ell -1,2}\right)
\\[4pt] 
\left( \ell +1\right) \left( \ell +2\left( \mu -\lambda -k\right) \right)
C_{2\ell +1,3} & = & \left( k-\ell \right) \left( \left( -1\right) ^{p\left(
F\right) }C_{2\ell ,1}-\left( k-\ell +2\lambda \right) C_{2\ell -1,3}\right)
\\[4pt] 
\left( \ell +1\right) \left( \ell +2\left( \mu -\lambda -k\right) \right)
C_{2\ell +1,4} & = & -\left( k-\ell +2\lambda \right) \left( \left(
-1\right) ^{p\left( F\right) }C_{2\ell ,4}+\left( k-\ell \right) C_{2\ell
-1,4}\right)
\end{array}
\right. 
\]

If $\mu -\lambda \neq 0,\frac{1}{2},1,\frac{3}{2},2...$, this system has
been solved and the solutions are the following:

\[
\left\{ 
\begin{array}{c}
C_{2\ell ,2}=\frac{\left( 
\begin{array}{c}
k-1 \\ 
\ell -1
\end{array}
\right) \left( 
\begin{array}{c}
k+2\lambda \\ 
\ell +1
\end{array}
\right) }{\left( 
\begin{array}{c}
-2\left( \mu -\lambda -k\right) \\ 
\ell +1
\end{array}
\right) } \\ 
C_{2\ell ,3}=-\frac{\left( 
\begin{array}{c}
k \\ 
\ell +1
\end{array}
\right) \left( 
\begin{array}{c}
k+2\lambda -1 \\ 
\ell -1
\end{array}
\right) }{\left( 
\begin{array}{c}
-2\left( \mu -\lambda -k\right) \\ 
\ell +1
\end{array}
\right) }
\end{array}
\right. and\left\{ 
\begin{array}{c}
C_{2\ell ,1}=\frac{\left( k+2\lambda -\ell \right) \left( 2\left( \mu
-\lambda -k\right) +\ell \right) }{\ell \left( k-\ell \right) }C_{2\ell ,3},
\\ 
C_{2\ell ,4}=-\frac{\left( k-\ell \right) \left( 2\left( \mu -\lambda
-k\right) +\ell \right) }{\ell \left( k-\ell +2\lambda \right) }C_{2l,2}, \\ 
C_{2\ell +1,1}=\left( -1\right) ^{p\left( F\right) }\frac{\left( k+2\lambda
-\ell \right) }{\ell }C_{2\ell ,3}, \\ 
C_{2\ell +1,2}=-\left( -1\right) ^{p\left( F\right) }\frac{\left( k-\ell
\right) }{\ell }C_{2\ell ,2}, \\ 
C_{2\ell +1,3}=\left( -1\right) ^{p\left( F\right) }\frac{\left( k+2\lambda
-\ell \right) }{\ell }C_{2\ell ,3} \\ 
C_{2\ell +1,4}=\left( -1\right) ^{p\left( F\right) }\frac{\left( k-\ell
\right) }{\ell }C_{2\ell ,2}
\end{array}
\right. . 
\]
That allows us to obtain the formula (\ref{StringQuantization}). 
\end{proof}

\subsection{Proof of theorem in case of $\left( k+\frac{1}{2}\right) $%
-differential operators}

\begin{proof}

In the case of $\left( k+\frac{1}{2}\right) $-differential operators where $%
k $ is integer, we get an $\mathrm{Aff}\left( 2|2\right) $-equivariant
quantization map by a straightforward calculation which is given by

\begin{eqnarray*}
Q\left( F_{1},F_{2}\right) &=&F_{1}\partial _{x}^{k}\overline{D}%
_{1}+F_{2}\partial _{x}^{k}\overline{D}_{2}+... \\
&&+\widetilde{Q}_{1}^{\left( 2\ell \right) }\left( F_{1}\right) +\widetilde{Q%
}_{1}^{\left( 2\ell +1\right) }\left( F_{1}\right) \\
&&+\widetilde{Q}_{2}^{\left( 2\ell \right) }\left( F_{2}\right) +\widetilde{Q%
}_{2}^{\left( 2\ell +1\right) }\left( F_{2}\right) \\
&&..+\left( C_{2k+1,1}\partial _{x}^{k}\overline{D}_{1}\left( F_{1}\right)
+C_{2k+1,2}\partial _{x}^{k}\overline{D}_{2}\left( F_{2}\right) \right)
\end{eqnarray*}
where the $\frac{m}{2}$-differential operators $\widetilde{Q}_{1}^{\left(
m\right) }$ and $\widetilde{Q}_{2}^{\left( m\right) }$ have the form :

\[
\left\{ 
\begin{array}{rcl}
\widetilde{Q}_{1}^{\left( 2\ell \right) }\left( F_{1}\right) & = & C_{2\ell
,1}\partial _{x}^{\ell }\left( F_{1}\right) \partial _{x}^{k-\ell }\overline{%
D}_{1}+C_{2\ell ,3}\partial _{x}^{\ell -1}\overline{D}_{1}\overline{D}%
_{2}\left( F_{1}\right) \partial _{x}^{k-\ell }\overline{D}_{2} \\[4pt] 
\widetilde{Q}_{1}^{\left( 2\ell +1\right) }\left( F_{1}\right) & = & 
C_{2\ell +1,1}\partial _{x}^{\ell }\overline{D}_{1}\left( F_{1}\right)
\partial _{x}^{k-\ell }+C_{2\ell +1,3}\partial _{x}^{\ell }\overline{D}%
_{2}\left( F_{1}\right) \partial _{x}^{k-\ell -1}\overline{D}_{1}\overline{D}%
_{2} \\[4pt] 
\widetilde{Q}_{2}^{\left( 2\ell \right) }\left( F_{2}\right) & = & C_{2\ell
,2}\partial _{x}^{\ell -1}\overline{D}_{1}\overline{D}_{2}\left(
F_{2}\right) \partial _{x}^{k-\ell }\overline{D}_{1}+C_{2\ell ,4}\partial
_{x}^{\ell }\left( F_{2}\right) \partial _{x}^{k-\ell }\overline{D}_{2} \\%
[4pt] 
\widetilde{Q}_{2}^{\left( 2\ell +1\right) }\left( F_{2}\right) & = & 
C_{2\ell +1,2}\partial _{x}^{\ell }\overline{D}_{2}\left( F_{2}\right)
\partial _{x}^{k-\ell }+C_{2\ell +1,4}\partial _{x}^{\ell }\overline{D}%
_{1}\left( F_{2}\right) \partial _{x}^{k-\ell -1}\overline{D}_{1}\overline{D}%
_{2}
\end{array}
\right. 
\]
The above quantization map commutes with the action of $X_{x^{2}}$ if and
only if the coefficients $C_{m,j}(j=1,2,3,4)$\ verify the following system
of linear equations :

\[
\left\{ 
\begin{array}{rcl}
\left( 
\begin{array}{c}
C_{2\ell ,1} \\ 
-\left( \ell +1\right) \left( \ell +2\left( 
\begin{array}{c}
\mu -\lambda \\ 
-k-\frac{1}{2}
\end{array}
\right) \right) C_{2\ell ,2}
\end{array}
\right) & = & \left( 
\begin{array}{c}
\left( k-\ell +1\right) \left( k-\ell +2\lambda +1\right) C_{2\ell -2,2} \\ 
+\left( -1\right) ^{p\left( F\right) }\left( k-\ell +1\right) C_{2\ell -1,2}
\\ 
+C_{2\ell ,4} \\ 
-\left( -1\right) ^{p\left( F\right) }\left( k-\ell +2\lambda +1\right)
C_{2\ell -1,4}
\end{array}
\right) \\[4pt] 
\left( 
\begin{array}{c}
\left( \ell +1\right) \left( \ell +2\left( 
\begin{array}{c}
\mu -\lambda \\ 
-k-\frac{1}{2}
\end{array}
\right) \right) C_{2\ell ,3} \\ 
+C_{2\ell ,4}
\end{array}
\right) & = & \left( 
\begin{array}{c}
C_{2l,1} \\ 
+\left( -1\right) ^{p\left( F\right) }\left( k-\ell +1\right) C_{2\ell -1,1}
\\ 
-\left( k-\ell +1\right) \left( k-\ell +2\lambda +1\right) C_{2\ell -2,3} \\ 
+\left( -1\right) ^{p\left( F\right) }\left( k-\ell +2\lambda +1\right)
C_{2\ell -1,3}
\end{array}
\right) \\[4pt] 
\left( 
\begin{array}{c}
\ell \left( \ell -1+2\left( \mu -\lambda -k-\frac{1}{2}\right) \right)
C_{2\ell ,1} \\ 
-C_{2\ell ,2}
\end{array}
\right) & = & -\left( k-\ell +1\right) \left( k-\ell +2\lambda +1\right)
C_{2\ell -2,1} \\[4pt] 
\left( 
\begin{array}{c}
C_{2\ell ,3} \\ 
+\ell \left( \ell -1+2\left( \mu -\lambda -k-\frac{1}{2}\right) \right)
C_{2\ell ,4}
\end{array}
\right) & = & -\left( k-\ell +1\right) \left( k-\ell +2\lambda +1\right)
C_{2\ell -2,4} \\[4pt] 
\left( 
\begin{array}{c}
C_{2\ell +1,1} \\ 
+\left( \ell +1\right) \left( \ell +2\left( 
\begin{array}{c}
\mu -\lambda \\ 
-k-\frac{1}{2}
\end{array}
\right) \right) C_{2\ell +1,2}
\end{array}
\right) & = & -\left( k-\ell +2\lambda \right) \left( 
\begin{array}{c}
\left( -1\right) ^{p\left( F\right) }C_{2\ell ,4} \\ 
+\left( k-\ell +1\right) C_{2\ell -1,2}
\end{array}
\right) \\[4pt] 
\left( 
\begin{array}{c}
C_{2\ell +1,2} \\ 
+\left( \ell +1\right) \left( \ell +2\left( 
\begin{array}{c}
\mu -\lambda \\ 
-k-\frac{1}{2}
\end{array}
\right) \right) C_{2\ell +1,1}
\end{array}
\right) & = & -\left( k-\ell +2\lambda \right) \left( 
\begin{array}{c}
\left( -1\right) ^{p\left( F\right) }C_{2\ell ,1} \\ 
+\left( k-\ell +1\right) C_{2\ell -1,1}
\end{array}
\right) \\[4pt] 
\left( 
\begin{array}{c}
C_{2\ell +1,3} \\ 
-\left( \ell +1\right) \left( \ell +2\left( 
\begin{array}{c}
\mu -\lambda \\ 
-k-\frac{1}{2}
\end{array}
\right) \right) C_{2\ell +1,4}
\end{array}
\right) & = & \left( k-\ell \right) \left( 
\begin{array}{c}
\left( k-\ell +2\lambda +1\right) C_{2\ell -1,4} \\ 
-\left( -1\right) ^{p\left( F\right) }C_{2\ell ,4}
\end{array}
\right) \\[4pt] 
\left( 
\begin{array}{c}
C_{2\ell +1,4} \\ 
-\left( \ell +1\right) \left( \ell +2\left( 
\begin{array}{c}
\mu -\lambda \\ 
-k-\frac{1}{2}
\end{array}
\right) \right) C_{2\ell +1,3}
\end{array}
\right) & = & \left( k-\ell \right) \left( 
\begin{array}{c}
\left( -1\right) ^{p\left( F\right) }C_{2\ell ,1} \\ 
+\left( k-\ell +2\lambda +1\right) C_{2\ell -1,3}
\end{array}
\right)
\end{array}
\right. 
\]
By solving this system, we obtain the formula (\ref{StringQuantization}). 
\end{proof}

\textbf{Acknowledgements}. I am very grateful to C.Duval, H. Gargoubi and
particularly to my supervisor V. Ovsienko for providing me with the problem
and the constant help. I am also thankful to D. Leites for his critical
reading of this paper and his helpful suggestions.

\end{document}